\documentclass{IEEEtran}

\usepackage{times}
\usepackage{amsmath}
\usepackage{graphicx}
\usepackage{amssymb}
\usepackage{subfigure}
\usepackage{amsthm}
\usepackage{algorithm}
\usepackage{algorithmic}
\usepackage{color}
\usepackage{subfig}
\usepackage{multirow}
\usepackage{url}
\usepackage{cite}
\usepackage[sort,square,numbers]{natbib}
\usepackage{stfloats}
\RequirePackage{CJKnumb}

\usepackage{microtype}

\usepackage{color}

\newtheorem{theorem}{Theorem}

\long\def\symbolfootnote[#1]#2{\begingroup
	\def\thefootnote{\fnsymbol{footnote}}\footnote[#1]{#2}\endgroup}

\title{New 
EVENODD+ Codes with More Flexible Parameters and Lower Complexity}

\author{Panyu Zhu$^\S$, Jingjie Lv$^\dagger$, Yunghsiang S. Han$^\S$, Linqi Song$^{\star}$, Hanxu Hou$^\dagger$$^{\star}$\\
$^\S$ Shenzhen Institute for Advanced Study, University of Electronic Science and Technology of China\\
$^\dagger$ School of Electrical Engineering \& Intelligentization, Dongguan University of Technology\\
 $^{\star}$ Department of Computer Science, City University of Hong Kong 
}

\begin{document}
\begin{sloppypar}
\begin{CJK}{UTF8}{gbsn}	
\maketitle
\thispagestyle{empty}
\vspace{-0.5cm}
\begin{abstract}\symbolfootnote[0]{This work was partially supported by the National Key Research and Development Program of China under Grant 2022YFA1004902, Key-Area Research and Development Program of Guangdong Province 2020B0101110003, National Key R\&D Program of
China (No. 2020YFA0712300), the National Natural Science Foundation of China (No. 62071121, 62371411),
Basic Research Enhancement Program of China under Grant 2021-JCJQ-JJ-0483.\\}
EVENODD+ codes are binary maximum distance separable (MDS) array codes for correcting double disk failures in RAID-6 with asymptotically optimal encoding/decoding/update complexities. However, the number of bits stored in each disk of EVENODD+ codes should be an odd number minus one.
In this paper, we present a new construction of EVENODD+ codes that have more flexible parameters. The number of bits stored in each disk of our codes is an odd minus one times any positive integer. Moreover, our codes not only have asymptotically optimal encoding/decoding/update complexities but also have lower encoding/decoding/update complexities than the existing EVENODD+ codes.
\end{abstract}
	
	
	\IEEEpeerreviewmaketitle
	
	\section{Introduction}
	
	\IEEEPARstart{A}RRAY codes were initially used to correct track errors in magnetic tape storage. After the concept of redundant arrays of inexpensive disk (RAID) \cite{RAID89} was proposed, such systems were characterized by a large number of data disks and a small number of redundant disks, with a focus on computational efficiency. Therefore, array codes are suitable for these systems \cite{I/O}, \cite{storage}. 
	
In a storage system, the more redundancy added to the system to increase fault tolerance, the higher the additional cost of the system will be. 
As a type of erasure coding, binary array codes \cite{array} have the advantages of simple operations and easy implementation. Maximum distance separable (MDS) array codes have the additional advantage of high storage efficiency. Encoding/decoding complexity is the key metric of binary MDS array codes that is defined as the total number of XORs involved in the encoding/decoding process. Another important metric is the update complexity, which is defined as the average number of parity
bits affected by a change of a single information bit. 
	
There are many binary MDS array codes in the literature. EVENODD codes \cite{blaum1995evenodd} and RDP codes \cite{corbett2004row} are two important binary MDS array codes with asymptotically optimal encoding/decoding complexity but with sub-optimal update complexity. Recently, EVENODD+ codes have been proposed in \cite{hou_new_2018} to achieve asymptotically optimal update complexity. However, the number of bits stored in each column is restricted to be an odd number minus one. In this paper, we will present new constructions of EVENODD+ that can support many more parameters.
Some other binary MDS array codes can be found in \cite{xu1999x,Plank2008The,Wu2011H,Li2012C,Shen2014HV,EEO,hou2014,Fu2016Short,Hou2018form,Hou2017A,extend+,hou2021star}.
	
	We first provide an overview of EVENODD codes \cite{blaum1995evenodd} and EVENODD+ codes \cite{hou_new_2018}.  EVENODD codes are $(p-1)\times (k+2)$ array codes, where $p\ge k$ is prime. The first $k$ columns are information columns, and the last two columns are parity columns.  For $i=0,1,\ldots,p-2$, let $b_{i,j}$ be the bits stored in column $j$, where $j=0,1,\ldots,k+1$.  The parity bits $b_{i,k}$ in column $k$ are computed by
	\[
	b_{i,k}=\sum_{j=0}^{k-1}b_{i,j},
	\]
	and the parity bits $b_{i,k+1}$ in column $k+1$ are computed by
	\[
	b_{i,k+1}=b_{p-1,k+1}+\sum_{j=0}^{k-1}b_{i-j,j},
	\]
	where $b_{p-1,j}=0$ for $j=0,1,\ldots,k-1$, and
	$b_{p-1,k+1}=\sum_{j=1}^{k-1}b_{p-1-j,j}$ is called the \emph{common bit}.  Note that the subscripts above are
	taken modulo $p$. 
The minimum update complexity of systematic MDS $(p-1)\times (k+2)$ array codes is $2+\frac{1}{p}(1-\frac{1}{k})$ \cite{Blaum1999On}.
We can count that the update complexity of EVENODD codes is $3-\frac{p+k-2}{k(p-1)}$, which is sub-optimal. The main reason for the sub-optimal update complexity is that the common bit $b_{p-1,k+1}$ is added to all the parity bits in column $k+1$.
EVENODD+ codes \cite{hou_new_2018} can achieve the asymptotically optimal update complexity by only adding the common bit $b_{p-1,k+1}$ to a partial of the parity bits in column $k+1$. Specifically, EVENODD+ codes add the common bit to the first $2\lfloor \frac{k}{2} \rfloor$ parity bits in column $k+1$.
The update complexity of EVENODD+ codes is $2+(2\lfloor \frac{k}{2} \rfloor-1)\frac{k-1}{k(p-1)}$, which is asymptotically optimal for $p\gg k$.

In this paper, We present a new construction for EVENODD+ codes with more flexible parameters. Our new codes have two properties: (i) MDS property and (ii) lower encoding/decoding/update complexity than EVENODD+ codes. The new codes build on the observation that the number of bits stored in each disk of EVENODD+ codes is an odd number minus one. In contrast, the number of bits stored in each disk of our codes can be an odd number minus one times any positive integer so as to support more parameters.
	

	\section{New Construction of EVENODD+ Codes}
	\label{sec:array_code}
In this section, we present a new construction of EVENODD+ codes with an array size of $\tau(p-1)\times(k+2)$, where $p$ is an odd number, $\tau$ and $k$ are positive integers. Denote the bit in column 
$j$ and row $i$ as $b_{i,j}$, where $j=0,1,\ldots,k+1$ and $i=0,1,\ldots,\tau(p-1)-1$. Columns $j$ with $j=0, 1,\ldots, k-1$ are information columns that store the information bits, and columns $k$ and $k+1$ are parity columns that store the parity bits. In the rest of the paper, the subscripts are taken modulo $\tau p$ unless otherwise specified.
	
As with EVEONDD codes,	the bits in column $k$ are computed by the summation of the information bits in the same row, i.e.,
	\begin{equation}
		b_{i,k}=\sum_{j=0}^{k-1}b_{i,j} \quad 0\leq i \leq \tau(p-1)-1.
		\label{framework:row}
	\end{equation}
For easier presentation, let 
$b_{i,j}=0$ for $i=\tau(p-1), \tau(p-1)+1,\ldots, \tau p-1$ and $j=0, 1,\ldots, k-1$. Let 
\[
t=\min(k-1,\tau).
\]
The bits in column  $k+1$ are computed by,
\begin{equation}
b_{i,k+1}=\begin{cases}
S_{i\bmod t}+&\sum\limits_{j=0}^{k-1}b_{i-j,j}  \\ 
& 0\le i\le 2\lfloor \frac{k-1}{2} \rfloor t-1,\\
\sum\limits_{j=0}^{k-1}b_{i-j,j} \\  
& 2\lfloor \frac{k-1}{2} \rfloor t\le i\le \tau(p-1)-1, \label{eq:2}
	\end{cases}
    \end{equation}
where $S_{\mu}=\sum\limits_{j=1}^{k-1}b_{\tau(p-1)+\mu-j,j}$ are $t$ common bits for $\mu=0,1,\dots, t-1$.
Similar to the construction of EVENODD+ codes in \cite{hou_new_2018}, we also add the common bit to a partial of the parity bits in column $k+1$. This is the essential reason for our codes to achieve asymptotically optimal update complexity. Note that the codes in \cite{hou_new_2018} are a special class of our codes with $\tau=1$.
Note that our codes can support many more parameters. The number of bits stored in each column of our codes can be an odd number minus one times any positive integer, while the number of bits in each column should be an odd number minus one in~\cite{hou_new_2018}.  

Table~\ref{table:evenodd} illustrates an example of our codes with $(\tau,p,k)=(2, 5, 3)$. We have two common bits $S_{0}=b_{7,1}+b_{6,2}$ and $S_{1}=b_{7,2}$ that are added to the first four bits in column~4.

 	\begin{table}[!htbp]
		\caption{\centering{Example of $(\tau,p,k)=(2, 5, 3)$.}}
		\label{table:evenodd}
		\centering
		\begin{tabular}{|c|c|c|c|c|} \hline
			$b_{0,0}$ & $b_{0,1}$ &  $b_{0,2}$& $b_{0,0}+b_{0,1}+b_{0,2}$& $b_{0,0}+S_{0}$ \\ \hline
			$b_{1,0}$ & $b_{1,1}$ &  $b_{1,2}$& $b_{1,0}+b_{1,1}+b_{1,2}$& $b_{1,0}+b_{0,1}+S_{1}$ \\ \hline
			$b_{2,0}$ & $b_{2,1}$ &  $b_{2,2}$& $b_{2,0}+b_{2,1}+b_{2,2}$& $b_{2,0}+b_{1,1}+b_{0,2}+S_{0}$ \\ \hline
			$b_{3,0}$ & $b_{3,1}$ &  $b_{3,2}$& $b_{3,0}+b_{3,1}+b_{3,2}$& $b_{3,0}+b_{2,1}+b_{1,2}+S_{1}$ \\ \hline
			$b_{4,0}$ & $b_{4,1}$ &  $b_{4,2}$& $b_{4,0}+b_{4,1}+b_{4,2}$& $b_{4,0}+b_{3,1}+b_{2,2}$ \\ \hline
			$b_{5,0}$ & $b_{5,1}$ &  $b_{5,2}$& $b_{5,0}+b_{5,1}+b_{5,2}$& $b_{5,0}+b_{4,1}+b_{3,2}$ \\ \hline
			$b_{6,0}$ & $b_{6,1}$ &  $b_{6,2}$& $b_{6,0}+b_{6,1}+b_{6,2}$& $b_{6,0}+b_{5,1}+b_{4,2}$ \\ \hline
			$b_{7,0}$ & $b_{7,1}$ &  $b_{7,2}$& $b_{7,0}+b_{7,1}+b_{7,2}$& $b_{7,0}+b_{6,1}+b_{5,2}$ \\ \hline
		\end{tabular}
		\vspace{-5pt}
	\end{table}

	\section{The MDS Property}
	\label{sec:mds}
	
In this section, we show that the proposed codes are MDS codes by giving the decoding method for any two-column failures.
	
 \begin{theorem}
Our codes are MDS codes if $p$ is an odd integer such that all divisors of p except 1 are larger than $k-1$.
\label{thm:mds}
\end{theorem}
\begin{proof}
Suppose that two columns are erased. We show that we can retrieve all the information bits if $p$ is an odd integer and all divisors of $p$ except 1 are larger than $k-1$. We divide the decoding method on two column erasures into two cases: $(i)$ one information and one parity column are erased; $(ii)$ two information columns are erased.

  
        
        
		
  
\textbf{Case $(i)$.} Suppose that columns~$f$ and $k+1$ are erased, where $0\leq f\leq k-1$.  We can recover the information bits in column $f$ by
		\[
		b_{i,k}+(b_{i,0}+b_{i,1}+\cdots+b_{i,f-1}+b_{i,f+1}+\cdots+b_{i,k-1})=b_{i,f},
		\]
according to Eq. \eqref{framework:row}, for $i=0,1,\ldots,\tau(p-1)-1$. Second, suppose that columns $f$ and $k$ are erased, where $0\leq f\leq k-1$. When $f=0$, we can compute the bit $b_{i,0}$ for $i=0,1,\ldots,\tau(p-1)-1$ by
     	\begin{align*}
     	&b_{i,k+1}+S_{i}+\sum_{j=1}^{k-1}b_{i-j,j}\\
    	={}&S_{i}+\sum\limits_{j=0}^{k-1}b_{i-j,j}+S_{i}+\sum_{j=1}^{k-1}b_{i-j,j}={}b_{i,0}.
        \end{align*}

When $1\leq f \leq k-1$, we have $0\le f-1 \le k-2 \le$2$\lfloor \frac{k}{2} \rfloor t-1$. According to Eq. \eqref{eq:2} with $i=f-1$, we can recover $b_{\tau(p-1)+(f-1)\bmod t-f,f}$ by
\begin{align*}
&b_{f-1,k+1}+\sum\limits_{j=1,j\neq f}^{k-1}b_{\tau(p-1)+(f-1)\bmod t-j,j}+\\
&\sum_{j=0,j\neq f}^{k-1}b_{f-1-j,j}\\
        ={}&\sum\limits_{j=0}^{k-1}b_{f-1-j,j}+\sum\limits_{j=1}^{k-1}b_{\tau(p-1)+(f-1)\bmod t-j,j}+\\
        &\sum\limits_{j=1,j\neq f}^{k-1}b_{\tau(p-1)+(f-1)\bmod t-j,j}+\sum_{j=0,j\neq f}^{k-1}b_{f-1-j,j}\\
        ={}&b_{\tau(p-1)+(f-1)\bmod t-f,f}+b_{p\tau-1,f}\\
        ={}&b_{\tau(p-1)+(f-1)\bmod t-f,f}.
        \end{align*}
Similarly, by Eq. \eqref{eq:2}, we can recover $b_{\tau(p-1)+i\bmod t-f,f}$ for $i=f-1,f-2,\ldots,\max(0,f-t)$, and then compute the $t$ common bits. For 
$$i\in \{0,1,\cdots,2\lfloor \frac{k-1}{2} \rfloor t-1\}\setminus \{f-1,f-2,\ldots,\max(0,f-t)\},$$ 
we can recover the bit $b_{i-f,f}$ by
		\begin{align*}
		&b_{i,k+1}+S_{i}+\sum_{j=0,j\neq f}^{k-1}b_{i-j,j}\\
		={}&S_{i}+\sum\limits_{j=0}^{k-1}b_{i-j,j}+S_{i}+\sum_{j=0,j\neq f}^{k-1}b_{i-j,j}={}b_{i-f,f}.
		\end{align*}
For $2\lfloor \frac{k-1}{2} \rfloor t\le i\le \tau(p-1)-1$, the common bit is not added to the bit $b_{i,k+1}$, we can recover the bit $b_{i-f,f}$ by
		\begin{align*}
		&b_{i,k+1}+\sum_{j=0,j\neq f}^{k-1}b_{i-j,j}\\
		={}&\sum\limits_{j=0}^{k-1}b_{i-j,j}+\sum_{j=0,j\neq f}^{k-1}b_{i-j,j}={}b_{i-f,f}.
		\end{align*}
We have recovered all the bits in column $f$.
		
\textbf{Case $(ii)$.} Suppose that information columns $f$ and $g$ are erased, where $0\leq f<g \leq k-1$.  We can first compute $\sum\limits_{\mu=0}^{t-1}S_{\mu}$ by summing all the parity bits in columns~$k$ and $k+1$. In the following, we present the decoding method by considering two sub-cases: $(ii.1)$ $\tau\geq k-1$ and $(ii.2)$ $\tau< k-1$. 
		
		
\textbf{Sub-case $(ii.1)$ $\tau\geq k-1$.} We have 
$t=k-1$ common bits.
		\label{>,g-f=a,ke}
We first consider the special case of $f=0$ and $g=k-1$. By subtracting the $\tau(p-1)(k-2)$ information bits in $k-2$
surviving information columns from the parity bits $b_{i,k}$ and $b_{i,k+1}$, we can obtain the following bits
\begin{equation}
b'_{i,k}=b_{i,0}+ b_{i,k-1},
\label{eq:3}
\end{equation}
where $i=0,1,\ldots,(p-1)\tau-1$,
\begin{equation}
b'_{i,k+1}=\begin{cases}
		S'_{i\bmod(k-1)}&+b_{i,0}+b_{i-(k-1),k-1},\\ 
& 0\le i\le 2\lfloor \frac{k-1}{2} \rfloor t-1,\\
		b_{i,0}&+b_{i-(k-1),k-1}, \\& 2\lfloor \frac{k-1}{2} \rfloor t\le i\le \tau(p-1)-1,
	    \end{cases}
\label{eq:4}
       \end{equation}
where $S'_{\mu}=b_{\tau(p-1)+\mu-(k-1),k-1}$ for $\mu=0,1,\ldots,k-2$.
If $\tau(p-1)$ is divisible by $k-1$, by summing the bit in Eq. \eqref{eq:3} with $i=0$ and the two bits in Eq. \eqref{eq:4} with $i=0$ and $i=k-1$, we have
\begin{align*}
&(S'_{0}+b_{0,0}+b_{1-k,k-1})+(S'_{0}+b_{k-1,0}+b_{0,k-1})+\\
&(b_{0,0}+b_{0,k-1})\\
={}&b_{k-1,0},
\end{align*}
where the above equality comes from $b_{1-k,k-1}=0$.
Then we can compute $b_{k-1,k-1}$ by Eq. \eqref{eq:3} with $i=k-1$. Let $i=\ell (k-1)$ with integer $\ell\geq 1$, the bit in Eq. \eqref{eq:4} does not contain the common bit if $\ell(k-1) \bmod \tau(p-1) \geq 2\lfloor \frac{k-1}{2} \rfloor(k-1)$, and we can compute $b_{\ell(k-1),0}$ and $b_{\ell(k-1),k-1}$ recursively by Eq. \eqref{eq:3} and Eq. \eqref{eq:4}, where $1\leq \ell \leq \frac{\tau(p-1)}{k-1}$. 

By summing the bit in  Eq. \eqref{eq:3} with $i=0+m$ and the two bits in Eq. \eqref{eq:4} with $i=0+m$ and $i=k-1+m$, we have
\begin{align*}
&(S'_{m}+b_{m,0}+b_{m+1-k,k-1})+(S'_{m}+b_{k-1+m,0}+b_{m,k-1})\\
&+(b_{m,0}+b_{m,k-1})\\
&={}b_{k-1+m,0},
\end{align*}
where $0 \le m \le k-2$, the above equality comes from $b_{m+1-k,k-1}=0$. Then we can compute $b_{k-1+m,k-1}$ by Eq. \eqref{eq:3} with $i=k-1+m$. By summing $b_{k-1+m,k-1}$ and the bit in Eq. \eqref{eq:4} with $i=2(k-1)+m$, we have
\begin{align*}
&b_{k-1+m,k-1}+(b_{2(k-1)+m,0}+b_{(k-1)+m,k-1})\\
={}&b_{2(k-1)+m,0}.
\end{align*}
Let $i=\ell (k-1)+m$ with integer $\ell\geq 1$. The bit in Eq. \eqref{eq:4} does not contain the common bit if $\ell(k-1)+m \bmod \tau(p-1) \geq 2\lfloor \frac{k-1}{2} \rfloor(k-1)$, and we can compute $b_{\ell(k-1)+m,0}$ and $b_{\ell(k-1)+m,k-1}$ recursively by Eq. \eqref{eq:3} and Eq. \eqref{eq:4}, where $1\leq \ell \leq \frac{\tau(p-1)}{k-1}$, $0 \le m \le k-2$. 

As the bit $b_{k-1+m,k-1}$ is known, if the bit in Eq. \eqref{eq:4} contains the common bit, we  use  Eq. \eqref{eq:4} with $i=2(k-1)+m$ to compute
\begin{align*}
&b_{k-1+m,k-1}+(S'_{m}+b_{2(k-1)+m,0}+b_{(k-1)+m,k-1})\\
={}&b_{2(k-1)+m,0}+S'_{m}.
\end{align*}
By Eq. \eqref{eq:3} with $i=2(k-1)+m$, we can compute $b_{2(k-1)+m,k-1}+S'_{m}$. By summing the bit in Eq. \eqref{eq:4} with $i=3(k-1)+m$, we have
\begin{align*}
&b_{2(k-1)+m,k-1}+S'_{m}+(S'_{m}+b_{3(k-1)+m,0}+b_{2(k-1)+m,k-1})\\
={}&b_{3(k-1)+m,0}.
\end{align*}
By repeating the above process for $\ell=2,3,\ldots,\frac{\tau(p-1)}{k-1}-1$, we can compute
\begin{align*}
&b_{\ell(k-1)+m,k-1}+S'_{m}\\
={}&b_{\tau(p-1)-(k-1)+m,k-1}+b_{\tau(p-1)+m-(k-1),k-1}\\
={}&0.
\end{align*}
Therefore, we can have $S'_{m}$, and further obtain $b_{\ell(k-1)+m,0}$ and $b_{\ell(k-1)+m,k-1}$ by Eq. \eqref{eq:3} and Eq. \eqref{eq:4}, where $1\leq \ell \leq \frac{\tau(p-1)}{k-1}$, $0 \le m \le k-2$. 
 
Otherwise, if $\ell \geq \frac{\tau(p-1)}{k-1}$, the bit in Eq. \eqref{eq:4} contains the common bit, we can compute one erased bit by summing two bits in Eq. \eqref{eq:4}. For example, when $\ell=\frac{\tau(p-1)}{k-1}$, by summing the known bit $b_{(\ell-1)(k-1),k-1}$ and the bit in Eq. \eqref{eq:4} with $i=\ell (k-1)$, we have 
\begin{align*}
&b_{(\ell-1)(k-1),k-1}+(b_{\ell(k-1),0}+b_{(\ell-1)(k-1),k-1}+S'_{0})\\
={}&b_{\ell(k-1),0}+b_{\tau(p-1)-(k-1),k-1}\\
	   	={}&b_{\ell(k-1),0}.
	   \end{align*}

When $i=\frac{\tau(p-1)}{k-1}$, we have $i(k-1)\bmod \tau(p-1)=0$, we can get 
	   \begin{align*}
	   	&b_{(i-1)(k-1),k-1}+[b_{i(k-1)\bmod\tau(p-1),0}+\\
            &b_{i(k-1)\bmod\tau(p-1)-(k-1),k-1}+S'_{i(k-1)\bmod\tau(p-1)\bmod k-1}]\\
	   	={}&b_{(i-1)(k-1),k-1}+[b_{i(k-1),0}+b_{p\tau-(k-1),k-1}+S'_{0}]\\
	   	={}&b_{i(k-1)-(k-1),k-1}+[b_{i(k-1),0}+0+b_{\tau(p-1)-(k-1),k-1}]\\
	   	={}&b_{i(k-1),0}.
	   \end{align*}
      
Next, we show that we have recovered all the erased bits, i.e., $\ell(k-1)+m=\{0,1,\ldots,\tau(p-1)-1\}$ with $\ell=1,2,\ldots,\frac{\tau(p-1)}{k-1}$ and $m=0,1,\ldots,k-2$. We need to show that $\ell_1(k-1)+m_1\neq \ell_2(k-1)+m_2$, for $1\leq \ell_1 < \ell_2 \leq \frac{\tau(p-1)}{k-1}$ and $0\leq m_1<m_2 \leq k-2$. Suppose that $\ell_1(k-1)+m_1=\ell_2(k-1)+m_2 \bmod \tau (p-1)$. We have
\begin{align*}
&\ell_1(k-1)+m_1+t\tau(p-1)=\ell_2(k-1)+m_2, \text{ and obtain,} \\  &	\frac{(\ell_2-\ell_1)(k-1)+(m_2-m_1)}{\tau(p-1)}=t,
      \end{align*}
where $t$ is a positive integer.
We can further obtain that
\begin{align*}
      	1\leq(\ell_2-\ell_1)\leq&\frac{\tau(p-1)}{k-1}-1,\\
      	k\leq(\ell_2-\ell_1)(k-1)+(m_2-m_1)\leq&\tau(p-1)-1,\\
      	t=\frac{(\ell_2-\ell_1)(k-1)+(m_2-m_1)}{\tau(p-1)}<&1,
      \end{align*}
which contradicts with $t\geq 1$.      Therefore, we have recovered all the erased bits of $(g,f)=(k-1,0)$ and $\tau(p-1)$ is divisible by $k-1$.
      
      \label{>,g-f=a,buke}

Next, we present the decoding method when $\tau(p-1)$ is not divisible by $k-1$.
By summing the bit in  Eq. \eqref{eq:3} with $i=0$ and the two bits in Eq. \eqref{eq:4} with $i=0$ and $i=k-1$, we have $b_{k-1,0}$. Then we can compute $b_{k-1,k-1}$ by Eq. \eqref{eq:3} with $i=k-1$. Let $i=\ell (k-1)$ with integer $\ell\geq 1$, the bit in Eq. \eqref{eq:4} does not contain the common bit if $\ell(k-1) \bmod \tau(p-1) \geq 2\lfloor \frac{k-1}{2} \rfloor(k-1)$, and we can compute $b_{\ell(k-1),0}$ and $b_{\ell(k-1),k-1}$ recursively by Eq. \eqref{eq:3} and Eq. \eqref{eq:4}, where $1\leq \ell \leq \tau(p-1)$. Otherwise, if the bit in Eq. \eqref{eq:4} contains the common bit, we can compute one erased bit and the common bit by summing two bits in Eq. \eqref{eq:4}. For example, when $\ell$ satisfies 0$\leq(\ell+1)(k-1)\bmod \tau(p-1)\leq k-2$, by summing the known bit $b_{\ell(k-1),k-1}$ and the bit in Eq. \eqref{eq:4} with $i=(\ell+1) (k-1)$, we have
      \begin{align*}
      	&b_{\ell(k-1),k-1}+(S'_{[(\ell+1)(k-1)\bmod\tau(p-1)]\bmod k-1}+\\
       &b_{(\ell+1)(k-1),0}+b_{[(\ell+1)(k-1)\bmod\tau(p-1)]-(k-1),k-1})\\
      	={}&b_{\ell(k-1),k-1}+b_{\ell(k-1),k-1}+\\
            & b_{(\ell+1)(k-1),0}+b_{p\tau-(k-1)+[(\ell+1)(k-1)\bmod\tau(p-1)],k-1}\\
      	={}&b_{\ell(k-1),k-1}+b_{\ell(k-1),k-1}+b_{(\ell+1)(k-1),0}+0\\
      	={}&b_{(\ell+1)(k-1),0}.
      \end{align*}
where the first equation comes from that $S'_{[(\ell+1)(k-1)\bmod\tau(p-1)]\bmod k-1}=b_{\ell(k-1),k-1}$. Therefore, we can have $S'_{[(\ell+1)(k-1)\bmod\tau(p-1)]\bmod k-1}$, and further obtain $b_{(\ell+1)(k-1),0}$ and $b_{(\ell+1)(k-1),k-1}$ by Eq. \eqref{eq:3} and Eq. \eqref{eq:4}, where $1\leq \ell \leq \tau(p-1)$. We have recovered all the erased information bits. 
       
       \label{>,g-f<a,ke and buke}
We have presented the decoding method for the special case of $f=0$ and $g=k-1$ in the above. In the following, we present the decoding method when $g-f < k-1$. 
First, we can compute the summation of the $t$ common bits by summing all the parity bits in columns~$k$
and $k+1$, i.e.,
\begin{eqnarray*}
&&\sum_{i=0}^{\tau(p-1)-1}b_{i,k}+\sum_{i=0}^{\tau(p-1)-1}b_{i,k+1} \nonumber \\
&=&\sum_{i=0}^{\tau(p-1)-1}\sum_{j=0}^{k-1}b_{i,j}+\sum_{i=0}^{\tau(p-1)-1}\sum_{j=0}^{k-1}b_{i-j,j}+\nonumber \\
&&\underbrace{S_{0}+\cdots+S_{t-1}+S_{0}+\cdots+S_{t-1}}_{(k-1)t \text{ terms if } k \text{ is odd}, kt \text{ terms if } k \text{ is even}} \\
&=&\sum_{j=0}^{k-1}\sum_{i=0}^{p\tau-1}b_{i,j}+\sum_{j=0}^{k-1}\sum_{i=0}^{p\tau-1}b_{i-j,j}+\sum_{j=0}^{k-1}\sum_{\mu=0}^{t-1}b_{\tau(p-1)+\mu-j,j}  \\
&=&\sum_{\mu=0}^{t-1}S_{\mu}. 
\end{eqnarray*}
By subtracting the $\tau(p-1)(k-2)$ information bits in $k-2$ surviving information columns from the parity bits $b_{i,k}$ and $b_{i,k+1}$, we can obtain the following bits
       \begin{equation}
       b'_{i,k}=b_{i,f}+ b_{i,g},
      \label{eq:5}
      \end{equation}
      where $i=0,1,\ldots,(p-1)\tau-1$,
      \begin{equation}
      b'_{i,k+1}=\begin{cases}
		S'_{i\bmod(k-1)}&+b_{i-f,f}+b_{i-g,g},\\ 
         & 0\le i\le 2\lfloor \frac{k-1}{2} \rfloor t-1,\\
		b_{i-f,f}&+b_{i-g,g}, \\& 2\lfloor \frac{k-1}{2} \rfloor t\le i\le \tau(p-1)-1,
	    \end{cases}
        \label{eq:6}
       \end{equation}
        where $i=0,1,\ldots,(p-1)\tau-1$ and $S'_{\mu}=b_{\tau(p-1)+\mu-f,f}+b_{\tau(p-1)+\mu-g,g}$ for $\mu=0,1,\ldots,k-2$. Similarly, we need to consider two cases: whether $\tau(p-1)$ is divisible by $(g-f)$. 
        If $\tau(p-1)$ is divisible by $(g-f)$, by summing the bit in  Eq. \eqref{eq:5} with $i=0-f, (g-f)-f,\dots, (k-3)(g-f)-f$, the bits in Eq. \eqref{eq:6} with $i=0, (g-f) ,\dots, (k-2)(g-f)$ and the sum of the common bits, we have
            \begin{align*}
    	&[S'_{0}+b_{p\tau-f,f}+b_{p\tau-g,g}+S'_{g-f}+b_{(g-f)-f,f}+b_{p\tau-f,g}+\\
            &\cdots+S'_{(k-2)(g-f)}+b_{(k-2)(g-f)-f,f}+b_{(k-3)(g-f)-f,g}]+\\
    	&[b_{p\tau-f,f}+b_{p\tau-g,g}+\cdots+b_{(k-3)(g-f)-f,f}+\\
     &b_{(k-3)(g-f)-f,g}]+\sum_{\mu=0}^{k-2}S'_{\mu}\\
    	={}&b_{(k-2)(g-f)-f,f},
    \end{align*}

By summing the bit in  Eq. \eqref{eq:5} with $i=0-f+m, (g-f)-f+m,\dots, (k-3)(g-f)-f+m$, the bits in Eq. \eqref{eq:6} with $i=0+m, (g-f)+m ,\dots, (k-2)(g-f)+m$ and the sum of the common bits, we have $b_{(k-2)(g-f)-f+m,f}$, where $0 \leq m \leq g-f-1$. Let $i=\ell (g-f)+m$ with integer $\ell\geq 1$. The bit in Eq. \eqref{eq:6} does not contain the common bit if $\ell(g-f)+m \bmod \tau(p-1) \geq 2\lfloor \frac{k-1}{2} \rfloor(k-1)$, and we can compute $b_{\ell(g-f)-f+m,f}$ and $b_{\ell(g-f)-f+m,g}$ recursively by Eq. \eqref{eq:5} and Eq. \eqref{eq:6}, where $k-2\leq \ell \leq \frac{\tau(p-1)}{g-f}+k-3$, $0 \le m \le g-f-1$. 
       
Otherwise, if the bit in Eq. \eqref{eq:6} contains the common bit, we can compute the common bit by the bit in Eq. \eqref{eq:6}. For example, when $\ell = \frac{\tau(p-1)}{g-f}$, the bit in Eq. \eqref{eq:6} with $i=\ell (g-f)+m$ is
    \begin{align*}
    	&S'_{\ell(g-f)+m}+b_{\ell(g-f)-f+m,f}+b_{\ell(g-f)-g+m,g}\\
            ={}&S'_{m}+b_{\tau p-f+m,f}+b_{\tau p-g+m,g}\\
    	={}&S'_{m},
    \end{align*}
where $0 \le m \le g-f-1$. Therefore, we can get $S'_{m}$, and further obtain $b_{\ell(g-f)-f+m,f}$ and $b_{\ell(g-f)-f+m,g}$ by Eq. \eqref{eq:5} and Eq. \eqref{eq:6}, where $k-2\leq \ell \leq \frac{\tau(p-1)}{g-f}+k-3$, $0 \le m \le g-f-1$.
  
When $\tau(p-1)$ is not divisible by $(g-f)$, by summing the bit in  Eq. \eqref{eq:5} with $i=0-f, (g-f)-f,\dots, (k-3)(g-f)-f$, the bits in Eq. \eqref{eq:6} with $i=0, (g-f) ,\dots, (k-2)(g-f)$ and the sum of the common bits, we have $b_{(k-2)(g-f)-f,f}$. Then let $i=\ell (g-f)$ with integer $\ell\geq 1$, we can similarly compute $b_{\ell(g-f)-f,f}$ and $b_{\ell(g-f)-f,g}$ if the bit in Eq. \eqref{eq:6} does not contain the common bit, where $k-2\leq \ell \leq$$\tau(p-1)+k-3$. If the bit in Eq. \eqref{eq:6} contains the common bit, we let $i=\ell (g-f)$ after getting $b_{(k-2)(g-f)-f,f}$. We can compute one erased bit and the common bit by summing two bits in Eq. \eqref{eq:6}. For example, when $\ell=\tau(p-1)-1$, by summing the known bit $b_{\ell(g-f),g}$ and the bit in Eq. \eqref{eq:6} with $i=(\ell+1) (g-f)$, we have
      \begin{align*}
      	&b_{\ell(g-f)-f,g}+(S'_{[(\ell+1)(g-f)\bmod\tau(p-1)]\bmod k-1}+\\
       &b_{(\ell+1)(g-f)-f,f}+b_{[(\ell+1)(g-f)\bmod\tau(p-1)]-g,g})\\
      	={}&b_{\ell(g-f)-f,g}+S'_{0}+b_{p\tau-f,f}+b_{p\tau-g,g}\\
      	={}&b_{p\tau-g,g}+S'_{0}+0+0\\
      	={}&S'_{0}.
      \end{align*}
Therefore, we can have $S'_{0}$, and further obtain $b_{\ell(g-f)-f,f}$ and $b_{\ell(g-f)-f,g}$ by Eq. \eqref{eq:5} and Eq. \eqref{eq:6}, where $k-2\leq \ell \leq$$\tau(p-1)+k-3$.

\textbf{Sub-case $(ii.2)$ $\tau< k-1$.} The decoding method of sub-case $(ii.2)$ $\tau< k-1$ is similar to that of sub-case $(ii.1)$ $\tau\geq k-1$. Please see the detailed proof in Appendix \ref{Appendix:Theorem2}.

\end{proof}
 
\vspace{-0.1cm}
\section{Complexity Analysis}
\vspace{-0.1cm}

    \begin{table*}[ht]
		\caption{The encoding/decoding/update complexity of our codes and EVENODD+ codes.}
		\label{table:complexity}
		\centering
        \renewcommand{\arraystretch}{2}
        \setlength{\tabcolsep}{1.5mm}{
		\begin{tabular}{|c|c|c|c|c|c|c|c|c|} \hline
			  \multicolumn{3}{|c|}{}  & \multicolumn{2}{|c|}{Encoding} & \multicolumn{2}{|c|}{Decoding} & \multicolumn{2}{|c|}{Update} \\ \hline
     
			\multicolumn{3}{|c|}{EVENODD+$(p,k)$} & \multicolumn{2}{|c|}{2$-\frac{2p-k}{k(p-1)}$} & \multicolumn{2}{|c|}{$2+\frac{2\lfloor\frac{k}{2}\rfloor-1}{k(p-1)}$} & \multicolumn{2}{|c|}{2$+(2\lfloor\frac{k}{2}\rfloor-1)\frac{k-1}{k(p-1)}$}\\ \hline
   
			\multicolumn{3}{|c|}{} &\multicolumn{2}{|c|}{~}&$(g-f)\mid \tau(p-1)$&$(g-f)\nmid\tau(p-1)$&\multicolumn{2}{|c|}{~}\\ \hline
   
			  \multirow{5}*{Our codes}&\multirow{2}*{$\tau \geq k-1$ }& $g-f=k-1$& \multicolumn{2}{|c|}{ \multirow{2}*{2$-\frac{k}{2}+\frac{2[\lfloor\frac{k-1}{2}\rfloor-1](k-1)}{k(p-1)\tau}$}}&$\frac{4}{k}+\frac{k-2}{k\tau(p-1)}$&$2+\frac{k-2}{k\tau(p-1)}$&\multicolumn{2}{|c|}{ \multirow{2}*{$2+(2\lfloor \frac{k-1}{2}\rfloor-1)\frac{k-1}{k\tau(p-1)}$}} \\ \cline{3-3} \cline{6-7}
     
                                      &&$g-f\textless k-1$&\multicolumn{2}{|c|}{~}&$\frac{4}{k}+\frac{(g-f)(k-2))}{k\tau(p-1)}$&$\frac{2(g-f+1)}{k}+\frac{(g-f)(k-2)}{k\tau(p-1)}$&\multicolumn{2}{|c|}{~}\\  \cline{2-9}
                                      
                                      &\multirow{3}*{$\tau \textless k-1$}&$g-f=\tau$&\multicolumn{2}{|c|}{ \multirow{3}*{2$-\frac{2p-2\lfloor\frac{k}{2}\rfloor-1}{k(p-1)}$}}&\multicolumn{2}{|c|}{$\frac{2}{k}+\frac{2p\tau-\tau-1}{k\tau(p-1)}$}&\multicolumn{2}{|c|}{ \multirow{3}*{$2+(2\lfloor \frac{k}{2}\rfloor-1)\frac{k-1}{k\tau(p-1)}$}}\\   \cline{3-3} \cline{6-7}
                                      
                                      &&$g-f\textless \tau$&\multicolumn{2}{|c|}{~}&\multirow{2}*{$\frac{4}{k}+\frac{(g-f)(\tau-1)-1}{k\tau(p-1)}$}&\multirow{2}*{$\frac{2(g-f+1)}{k}+\frac{(g-f)(\tau-1)}{k\tau(p-1)}$}&\multicolumn{2}{|c|}{~}\\   \cline{3-3} 
                                      
                                      &&$g-f\textgreater \tau$&\multicolumn{2}{|c|}{~}&&&\multicolumn{2}{|c|}{~}\\  \hline

		\end{tabular}}
		\vspace{-5pt}
     \end{table*}

In this section, we evaluate encoding/decoding/update complexities for our codes. Table \ref{table:complexity} summarizes the results of our codes and EVENODD+ codes \cite{hou_new_2018}. We focus on the decoding complexity for decoding two information erasures.

We first consider the encoding complexity. Computing the bits in column~$k$ takes $(k-1)\tau(p-1)$ XORs. Computing the bits in column~$k+1$ takes $(k-1)\tau(p-1)-(k-1)+2\lfloor \frac{k-1}{2} \rfloor(k-1)$ XORs when $\tau\geq k-1$ and $(k-1)\tau(p-1)-\tau+2\lfloor \frac{k}{2} \rfloor\tau$ XORs when $\tau < k-1$.  Thus, the encoding complexity is $2[(k-1)\tau(p-1)]-(k-1)+2\lfloor \frac{k-1}{2} \rfloor(k-1)$ when $\tau\geq k-1$ and $2[(k-1)\tau(p-1)]-\tau+2\lfloor \frac{k}{2} \rfloor\tau$ when $\tau < k-1$.
	
Next, we consider the decoding complexity of two information erasures. In our decoding method, we first compute the summation of the $t$ common bits that takes $2\tau(p-1)-1$ XORs. 
Then, we divide the decoding method into two cases: $\tau\geq k-1$ and $\tau < k-1$. When $\tau\geq k-1$ and $g-f=k-1$, if $\tau(p-1)$ is divisible by $g-f$, we can compute the erased bits with $(k-1)[1+\frac{2\tau(p-1)}{k-1}]$ XORs; otherwise, if $\tau(p-1)$ is not divisible by $g-f$, then it requires $(k-1)[1+2\tau(p-1)]$ XORs. When $\tau\geq k-1$ and  $g-f<k-1$, we require $(g-f)[k-1+\frac{2\tau(p-1)}{g-f}]$ if $\tau(p-1)$ is divisible by $g-f$ and $(g-f)[k-2+2\tau(p-1)]$ XORs if $\tau(p-1)$ is not divisible by $g-f$. 
	
When $\tau < k-1$ and $g-f=\tau$, we can obtain the bits with $\tau[1+2(p-1)]$ XORs. When $\tau < k-1$ and $g-f<\tau$, we require $(g-f)[\tau-1+\frac{2\tau(p-1)}{g-f}]$ XORs if $\tau(p-1)$ is divisible by $g-f$ and $(g-f)[\tau-1+2\tau(p-1)]$ XORs if $\tau(p-1)$ is not divisible by $g-f$. When $\tau < k-1$ and  $g-f > \tau$, we require $(g-f)[\tau-1+\frac{2\tau(p-1)}{g-f}]$ XORs if $\tau(p-1)$ is divisible by $g-f$ and $(g-f)[\tau-1+2\tau(p-1)]$ XORs if $\tau(p-1)$ is not divisible by $g-f$. 
	
We define the \emph{normalized encoding complexity} as the ratio of encoding complexity to the number of information bits and \emph{normalized decoding complexity} as the ratio of decoding complexity to the number of information bits. The normalized encoding complexity of our codes is $2-\frac{2}{k}$ when $\tau\geq k-1$ and $2-\frac{2p-2\lfloor \frac{k}{2} \rfloor-1}{k(p-1)}$ when $\tau \leq k$. Recall that the normalized encoding complexity of EVENODD+($p, k$) \cite{hou_new_2018} is $2-\frac{2p-k}{k(p-1)}$. We can see that when $\tau\geq k-1$, the normalized encoding complexity of our codes is smaller than that of EVENODD+ codes. The normalized decoding complexity of EVENODD+($p, k$) \cite{hou_new_2018} is $2+\frac{2\lfloor \frac{k}{2} \rfloor-1}{k(p-1)}$. From the results in Table \ref{table:complexity}, we see that our codes have lower decoding complexity than EVENODD+ codes when $\tau\geq k-1$.

Finally, we consider the update complexity. 
If an information bit is changed, we need to update one parity bit in column $k$ and $1+(2\lfloor \frac{k-1}{2}\rfloor-1)\frac{k-1}{k\tau(p-1)}$ parity bits in column $k+1$ on average when $\tau \geq(k-1)$ and update $1+(2\lfloor \frac{k}{2}\rfloor-1)\frac{k-1}{k\tau(p-1)}$ parity bits when $\tau <(k-1)$. Thus, the update complexity of our codes is $2+(2\lfloor \frac{k-1}{2}\rfloor-1)\frac{k-1}{k\tau(p-1)}$ when $\tau \geq(k-1)$ and $2+(2\lfloor \frac{k}{2}\rfloor-1)\frac{k-1}{k\tau(p-1)}$ when $\tau <(k-1)$. Recall that the update complexity of EVENODD+ codes is $2+(2\lfloor \frac{k}{2} \rfloor-1)\frac{k-1}{k(p-1)}$. Therefore, our codes have lower update complexity than EVENODD+ codes.

\section{Conclusion}
\label{sec:discussions}
In this paper, we present a new construction for EVENODD+ codes with two parity columns such that the number of bits stored in each column is prime minus one times any positive integer. We show that our codes are MDS codes. Moreover, our codes have lower encoding/decoding/update complexity than the existing EVENODD+ codes. 
	
	   \bibliographystyle{IEEEtran}
	   \bibliography{document.bib}
\end{CJK}
\clearpage

\appendices

\section{Proof of \textbf{Sub-case $(ii.2)$ $\tau< k-1$} in Theorem \ref{thm:mds}}\label{Appendix:Theorem2}

\begin{IEEEproof}
When $\tau < k-1$, the number of common bits is $ \tau$, we can divide the condition $0\leq f<g\leq k-1$ into two sub-conditions of $0\leq f<g\leq \tau$ and $\tau<g\leq k-1$, where the decoding method of the first sub-condition is the same as that of $\tau< k-1$. We only need to consider the condition of $\tau < g \leq k-1$.

        \label{<,g-f=a}
As with the sub-case $(ii.1)$ $\tau\geq k-1$, we still first consider the special sub-case of $g-f=\tau$. Suppose that $\tau(p-1)$ is divisible by $\tau$, and we can obtain the following bits by subtracting the $\tau(p-1)(k-2)$ information bits in $k-2$ surviving information columns from the parity bits $b_{i,k}$ and $b_{i,k+1}$,
        \begin{equation}
        	b'_{i,k}=b_{i,f}+ b_{i,g},
         \label{eq:7}
        \end{equation}
        where $i=0,1,\ldots,(p-1)\tau-1$.
        \begin{equation}
        	b'_{i,k+1}=\begin{cases}
        		S'_{i(\bmod\tau)}+&b_{i-f,f}+b_{i-g,g}, \\
          &0\le i\le 2\lfloor \frac{k}{2} \rfloor\tau-1,\\
        		b_{i-f,f}+&b_{i-g,g},\\
          &2\lfloor \frac{k}{2} \rfloor\tau\le i\le \tau(p-1)-1,
        	\end{cases}
         \label{eq:8}
        \end{equation}
where $i=0,1,\ldots,(p-1)\tau-1$ and $S'_{\mu}=b_{\tau(p-1)+\mu-f,f}+b_{\tau(p-1)+\mu-g,g}$ for $\mu=0,1,\ldots,\tau-1$. 
When $g-f=\tau$, we can view the $\tau(p-1)\times (k+2)$ array of our codes as the $\tau$ codewords of EVENODD+ codes. Therefore, we can recover all the erased bits in columns $g$ and $f$ by the decoding method of EVENODD+ codes in \cite{hou_new_2018}, under the condition that $p$ is an odd number such that all divisors of $p$ except 1 are larger than $k-1$.

    \label{<,g-f<a,ke}
In the following, we consider the case of $g-f<\tau$, we can get that $2\leq f<g$ and $\tau < g \leq k-1$. In this situation, we still need to discuss whether $\tau(p-1)$ is divisible by $g-f$. When $\tau(p-1)$ is divisible by $g-f$, we sum the bit in  Eq. \eqref{eq:7} with $i=0, (g-f) ,\dots, (\tau-2)(g-f)$, the bits in Eq. \eqref{eq:8} with $f, g, 2g-f ,\dots, (\tau-1)g-(\tau-2)f$ and the summation of the $\tau$ common bits to obtain
    \begin{align*}
    	&[S'_{f}+b_{0,f}+b_{f-g,g}+S'_{g}+b_{g-f,f}+b_{0,g}+\cdots+\\
            &S'_{(\tau-1)g-(\tau-2)f}+b_{(\tau-1)g-(\tau-1)f,f}+b_{(\tau-2)g-(\tau-2)f,g}]+\\
    	&[b_{0,f}+b_{0,g}+\cdots+b_{(\tau-2)(g-f),f}+b_{(\tau-2)(g-f),g}]+\sum_{\mu=0}^{\tau-1}S'_{\mu}\\
    	={}&b_{(\tau-1)(g-f),f},
    \end{align*}
    where $b_{p\tau+f-g}=0$.
    Let $i=(\ell+1)g-\ell f+m$ with integer $\ell\geq 1$. When the bit in Eq. \eqref{eq:8} does not contain the common bit,  similarly, we can calculate $b_{\ell(g-f)+m,f}$ and $b_{\ell(g-f)+m,g}$ recursively by Eq. \eqref{eq:7} and Eq. \eqref{eq:8}, where $(\tau-1)\leq \ell \leq$$\frac{\tau(p-1)}{g-f}+(\tau-2)$, 0$\leq m \leq (g-f-1)$.
     If the bit in Eq. \eqref{eq:8} contain the common bit, when $\ell=\frac{\tau(p-1)}{g-f}-1$, we can find the bits in Eq. \eqref{eq:8} with $i=(\ell+1)g-\ell f$ have
     $$S'_{(\ell+1)g-\ell f}+b_{[(\ell+1)g-\ell f]-f,f}+b_{[(\ell+1)g-\ell f]-g,g}.$$
     Since  $[(\ell+1)g-\ell f]\bmod \tau(p-1)=f$, the above bit is the common bit $S'_{f}$, so we can compute $b_{\ell(g-f)+m,f}$ and $b_{\ell(g-f)+m,g}$, for $(\tau-1)\leq \ell \leq$$\frac{\tau(p-1)}{g-f}+(\tau-2)$, 0$\leq m \leq (g-f-1)$.
     
    \label{<,g-f<a,buke}
    When $\tau(p-1)$ is not divisible by $g-f$, by summing the bit in  Eq. \eqref{eq:7} with $i=0, (g-f) ,\dots, (\tau-2)(g-f)$, the bits in Eq. \eqref{eq:8} with $f, g, 2g-f ,\dots, (\tau-1)g-(\tau-2)f$ and the sum of the common bits, we have $b_{(\tau-1)(g-f),f}$. Let $i=(\ell+1)g-\ell f$ with integer $\ell\geq 1$. When the bit in Eq. \eqref{eq:8} does not contain the common bit, we can calculate $b_{\ell(g-f),f}$ and $b_{\ell(g-f),g}$ recursively by Eq. \eqref{eq:7} and Eq. \eqref{eq:8}, where $(\tau-1)\leq \ell \leq$$\tau(p-1)+(\tau-2)$. Otherwise, if the bit in Eq. \eqref{eq:8} contains the common bit, let $\ell=\tau(p-1)-1$, the bit with $i=(\ell+1)g-\ell f$ in Eq. \eqref{eq:8} have
    $$S'_{(i+1)g-if}+b_{[(i+1)g-if]mod\tau(p-1)-f,f}+b_{[(i+1)g-if]mod\tau(p-1)-g,g},$$
    we can obtain the common bit and further compute the information bits $b_{\ell(g-f),f}$ and $b_{\ell(g-f),g}$, where $\tau-1 \leq \ell \leq$$\tau(p-1)+\tau-2$.
    
    \label{<,g-f>a,ke,buke}
Consider that $g-f > \tau$. if $\tau(p-1)$ is divisible by $g-f$, we sum the bit in  Eq. \eqref{eq:7} with $i=0, (g-f) ,\dots, (\tau-2)(g-f)$, the bits in Eq. \eqref{eq:8} with $i=f, g, 2g-f ,\dots, (\tau-1)g-(\tau-2)f$ and the summation of common bits to obtain
    \begin{align*}
    	&[S'_{f}+b_{0,f}+b_{f-g,g}+S'_{g}+b_{g-f,f}+b_{0,g}+\cdots+\\
    	&S'_{(\tau-1)g-(\tau-2)f}+b_{(\tau-1)g-(\tau-2)f-f,f}+b_{(\tau-1)g-(\tau-2)f-g,g}+\\
    	&[b_{0,f}+b_{0,g}+\cdots+b_{(\tau-2)(g-f),f}+b_{(\tau-2)(g-f),g}]+\sum_{\mu=0}^{\tau-1}S'_{\mu}\\
    	={}&b_{p\tau-(g-f),g}+b_{(\tau-1)(g-f),f}.
    \end{align*}
     Let $\ell= \frac{\tau(p-1)}{g-f}-1$, by summing the bit in Eq. \eqref{eq:8} with $i=(\ell+1)g-\ell f$ and $b_{\ell(g-f),f}+b_{p\tau-(g-f),g}$, we have
     \begin{align*}
     	&b_{\ell(g-f),f}+b_{p\tau-(g-f),g}+S'_{(\ell+1)g-\ell f}+\\
     	&b_{[(\ell+1)g-\ell f]\bmod\tau(p-1)-f,f}+b_{[(\ell+1)g- \ell f]\bmod\tau(p-1)-g,g}\\
     	={}&b_{\tau(p-1)-(g-f),f}+b_{p\tau-(g-f),g}+b_{\tau(p-1)+[(\ell+1)g-\ell f]-f,f}+\\
     	&b_{\tau(p-1)+[(\ell+1)g-\ell f]-g,g}+b_{[(\ell+1)g-\ell f]\bmod\tau(p-1)-f,f}+\\
            &b_{p\tau-(g-f),g}\\
     	={}&b_{[(\ell+1)g-\ell f]\bmod\tau(p-1)-f,f}={}b_{0,f}.
     \end{align*}
Therefore, we get the bit $b_{p\tau-(g-f),g}$ and the common bit $S'_{(\ell+1)g-\ell f}$. Let $i=(\ell+1)g-\ell f+m$ with integer $\ell\geq 1$, we can obtain $b_{\ell(g-f)+m,f}$ and $b_{\ell(g-f)+m,g}$ recursively by Eq. \eqref{eq:7} and Eq. \eqref{eq:8}, where $(\tau-1)\leq \ell \leq$$\frac{\tau(p-1)}{g-f}+(\tau-2)$, 0$\leq m \leq (g-f-1)$.
     
If $\tau(p-1)$ is not divisible by $g-f$, we can get $b_{p\tau-(g-f),g}+b_{(\tau-1)(g-f),f}$ similarly. When the bit in Eq. \eqref{eq:8} contains the common bit, let $i=(\ell+1)g-\ell f$ with integer $\ell\geq 1$. When $\ell=\tau(p-1)-1$, the bit $b_{\ell(g-f),g}+S'_{(\ell+1)g-\ell f}+b_{p\tau-(g-f),g}$ we can have
    \begin{align*}
    	&b_{\ell(g-f),g}+S'_{(\ell+1)g-\ell f}+b_{p\tau-(g-f),g}\\
    	={}&b_{[\tau(p-1)(g-f)-(g-f)]\bmod\tau(p-1),g}+S'_{(\ell+1)g-\ell f}+b_{p\tau-(g-f),g}\\
    	={}&b_{p\tau-(g-f),g}+S'_{(\ell+1)g-\ell f}+b_{p\tau-(g-f),g}={}S'_{(\ell+1)g-\ell f}.
    \end{align*}
     We can get $S'_{(\ell+1)g-\ell f}$. After obtaining the common bit, we can compute $b_{\ell(g-f)},f$ and $b_{\ell(g-f)},f$ by Eq. \eqref{eq:7} and Eq. \eqref{eq:8}, where $(\tau-1)\leq \ell \leq$$\tau(p-1)+(\tau-2)$.
     
We have shown that we can always compute all the erased bits for any two erased columns, and therefore, our codes are MDS codes.

\end{IEEEproof}
\end{sloppypar}
 \end{document}